\documentclass[preprint,12pt]{elsarticle}
\usepackage{etoolbox}
\makeatletter
\def\ps@pprintTitle{%
	\let\@oddhead\@empty
	\let\@evenhead\@empty
	\def\@oddfoot{\reset@font\hfil\thepage\hfil}
	\let\@evenfoot\@oddfoot
}
\makeatother
\usepackage{graphics}
\usepackage{amsmath, amsthm, amssymb}
\usepackage{natbib}
\usepackage[colorlinks=true]{hyperref}
\usepackage{lscape}
\usepackage[utf8]{inputenc} 
\usepackage{amsthm}
\usepackage{graphicx}
\usepackage{subcaption}

\usepackage{enumerate}
\usepackage{mathrsfs}
\usepackage{setspace}
\usepackage{multirow}
\usepackage{graphicx}
\usepackage{amsfonts}
\usepackage{verbatim}
\usepackage{amssymb}
\usepackage{amsthm}
\usepackage{multirow,bigstrut,bigdelim}
\newtheorem{theorem}{Theorem}[section]
\newtheorem{lem}{Lemma}[section]

\theoremstyle{plain}
\newtheorem{definition}{Definition}[section]

\newtheorem{example}{Example}[section]

\theoremstyle{remark}
\newtheorem{rem}{Remark}[section]
\newtheorem{corollary}{Corollary}[section]
\numberwithin{equation}{section}
\usepackage[mathscr]{euscript}
\usepackage{geometry} 
\usepackage{graphicx,lscape} 
\usepackage{booktabs}
\usepackage{array}
\usepackage{paralist} 
\usepackage{verbatim}
\usepackage{subcaption}
\biboptions{comma,round,authoryear}

\begin{document}
\begin{frontmatter}
	
		\title{\textbf{Further results on relative, divergence measures based on extropy and their applications}}
        \author{Saranya P.$ ^* $, S.M.Sunoj}
		\ead{smsunoj@cusat.ac.in, saranyapanat96@gmail.com}
		\cortext[cor1]{Corresponding author}
		\address{Department of Statistics\\Cochin University of Science and Technology\\Cochin 682 022, Kerala, INDIA.}
\begin{abstract}
This study explores information measures based on extropy, introducing dynamic relative extropy measures for residual and past lifetimes, and investigating their various properties. Furthermore, the study analyzes the relationships between extropy-based divergence with dynamic relative extropy and other extropy measures. A nonparametric estimator for relative extropy is developed, and its performance is assessed through numerical simulation studies. The practical applicability of the relative extropy is demonstrated through some real-life data sets. 
\end{abstract}	
\end{frontmatter}
\section{Introduction}
In scientific experiments, probability assessments are often compromised by limited or inaccurate data, affecting statistical estimation and inference. To address this, \cite{kerridge1961inaccuracy} introduced the Kerridge measure, which quantifies the discrepancy between two probability functions. For two absolutely continuous random variables $X$ and $Y$ with CDFs 
$F$ and $G$, and PDFs $f$ and $g$, respectively, the Kerridge inaccuracy measure is defined as
\[
I(F, G) = -\int_{0}^{\infty} {f(x) \log g(x) dx}.
\]
This measure can also be expressed as the sum of Shannon entropy $H(F)$ (see \cite{shannon1948mathematical}) and Kullback-Leibler (KL) divergence (\cite{kullback1951information}) as, 
\[
I(F, G) = H(F) + K(F, G),
\]
 where $H(F) = -\int_{0}^{\infty}{f(x) \log f(x) dx}$ is the Shannon entropy and KL divergence is denoted by $K(F, G) = \int_{0}^{\infty}{f(x) \log \frac{f(x)}{g(x)}dx}$. Note that the inaccuracy measure reduces to the corresponding uncertainty measure when the distributions are identical.
More recently, extropy has emerged as a complementary dual to entropy, offering a novel perspective on uncertainty quantification. According to \cite{lad2015extropy}, the extropy of the random variable $X$ is defined as
\begin{equation}\label{1.1}
    J(X)=-\frac{1}{2}\int_0^{\infty} f^2(x) dx.
\end{equation}
\cite{lad2015extropy} also defined the concept of relative divergence measure based on extropy, as
\begin{equation}\label{eq:R4relative extropy}
   d(f,g)=\frac{1}{2}\int_0^\infty \left({f(x)}-{g(x)}\right)^2 dx. 
\end{equation}   
\cite{hashempour2024dynamic} introduced an extropy-based dynamic cumulative past inaccuracy measure and studied its characteristic properties. Also, \cite{hashempour2024new} provided a new measure of inaccuracy based
on extropy for record statistics between distributions of the $n$th upper (lower) record value
and parent random variable which is given in the form,
\begin{equation}\label{1.3}
    \xi J(X,Y)=-\frac{1}{2}\int_0^\infty f(x)g(x)dx.
\end{equation}
\cite{toomaj2023extropy} derived some properties and several theoretical merits of extropy, and dynamic versions of extropy.  The dynamic extropy inaccuracy for residual lifetimes, $X_t = (X - t| X > t)$ and $Y_t = (Y - t| Y > t)$, defined by \cite{hashempour2024residual} as 
\begin{equation}
    \xi J_r(X,Y,t)=-\frac{1}{2}\int_t^\infty \frac{f(x)}{\bar{F}(t)}\frac{g(x)}{\bar{G}(t)}dx,
\end{equation}
equals to $J_t(X)$ when $f(x)=g(x)$.
    Let $X_{(t)} = (t - X | X \leq t)$ and $Y_{(t)} = (t - Y | Y \leq t)$ be two past lifetime random variables, then  dynamic past inaccuracy measure defined by \cite{mohammadi2024dynamic} is given as
    \begin{equation}\label{eq:R4pastinaccuracy}
        \xi J_p(X,Y,t)=-\frac{1}{2}\int_0^t\frac{f(x)g(x)}{F(t)G(t)}dx.
    \end{equation}
The discrimination information based 
on extropy and inaccuracy between density functions $f(x)$ and $g(x)$ is defined as
\begin{equation}\label{eq:R3discrimeasuredensity}
    J(f|g)=\frac{1}{2}\int_0^\infty [f(x)-g(x)]f(x)dx = \frac{1}{2}E_f[f(X)-g(X)].
\end{equation}
For nonnegative random variables $X$ and $Y$, the dynamic extropy divergence between two distributions $f$ and $g$ is defined due to \cite{mohammadi2024dynamic}, given by
    \begin{equation}
       J_r(f_t|g_t)=\frac{1}{2}\int_t^\infty \left(\frac{f(x)}{\bar{F}(t)}-\frac{g(x)}{\bar{G}(t)}\right) \frac{f(x)}{\bar{F}(t)}dx,
    \end{equation}
   and the dynamic past extropy divergence is defined by 
    \begin{equation}
        J_p(f_t|g_t)=-\frac{1}{2}\int_0^t\left(\frac{f(x)}{F(t)}-\frac{g(x)}{G(t)}\right)\frac{f(x)}{F(t)}dx.
    \end{equation}
\cite{saranya2024relative} proposed a new relative extropy measure based on the survival function, studied its properties, usefulness in testing uniformity, and its applications. Recently, \cite{saranya2025inaccuracy} introduced a new cumulative extropy measure of inaccuracy and divergence, and obtained their various properties and applications in image classification and reliability analysis.  Although the literature includes symmetric and asymmetric information measures based on entropy and extropy, each with distinct advantages and applications, the interconnections between these measures remain relatively unexplored. Investigating these relationships can provide deeper insight into their properties and lead to significant results in information theory and reliability. Accordingly, in the present study we focus on exploring extropy-based information measures.

The paper is organized as follows. In Section 2, we discuss some properties of relative extropy and extropy divergence.  We explore the dynamic relative extropy for residual lifetimes and derive some results and properties along with those of dynamic extropy divergences and dynamic extropy inaccuracy in Section 3. The past lifetime scenario of these measures is discussed in Section 4. Section 4 investigates the estimation and real data analysis of the relative extropy. 
\section{Relative extropy and extropy divergence}
In this section, we study some properties of relative extropy and extropy divergence. For two nonnegative continuous random variables $X$ and $Y$ with pdfs $f$ and $g$ respectively, using \eqref{1.1}, \eqref{eq:R4relative extropy} and \eqref{1.3}, we have the relationship between relative extropy, extropy inaccuracy, and extropy as: 
 \[
 d(f,g)=2\xi J(X,Y)-J(Y)-J(X).
 \]
 The relative extropy is always nonnegative and zero if $f(x)=g(x)$ for every $x$. The relationship between relative extropy and extropy divergence is given by the following theorem.
\begin{theorem}\label{Theorem:R4splitRE}
    Let $X$ and $Y$ be two nonnegative random variables with density functions $f$ and $g$ respectively. Then the relative extropy of $X$ and $Y$ can be defined as the sum of extropy divergences as follows.
    \begin{equation}
        d(f,g)=J(f|g)+J(g|f)=\frac{1}{2}(E_f[f(x)-g(x)]+E_g[g(x)-f(x)]),
    \end{equation}
    where $J(g|f)$ is the extropy divergence from $g$ to $f$ with respect to $g$ given as follows:
    \begin{equation*}
     J(g|f)=\frac{1}{2}\int_0^\infty [g(x)-f(x)]g(x)dx.   
    \end{equation*}
\end{theorem}
\begin{proof}
 We can split \eqref{eq:R4relative extropy} as follows:
\begin{equation*}
    \begin{split}
     d(f,g)&=\frac{1}{2}\int_0^\infty \left({f(x)}-{g(x)}\right)^2 dx \\
     &=\frac{1}{2}\int_0^\infty f^2(x)-2f(x)g(x)+g^2(x) dx\\
     &=\frac{1}{2}\left(\int_0^\infty ({f(x)}-{g(x)})f(x) dx+ \int_0^\infty ({g(x)}-{f(x)}) g(x) dx\right)\\
     &=J(f|g)+J(g|f)
    \end{split}
\end{equation*}   
\end{proof}
\begin{rem}
    $J(f|g)$ and $J(g|f)$ can be both nonnegative and nonpositive, however, their sum is always nonnegative. 
\end{rem}

\cite{kullback1959} studied the approximation of the Kullback-Leibler information of $f(x)$ and $f(x;\theta+\Delta\theta)$. In the following theorem, we provide a theoretical approximation to relative extropy of two similar densities, that is, between $f(x,\theta)$ and $f(x,\theta+\Delta\theta)$ as follows:
\begin{theorem}
    Let $X$ be a nonnegative random variable with density function $f$ with parameter $\theta$. Then, 
    \begin{equation}\label{R4:eqapprox}
        d(f(x,\theta),f(x,\theta+\Delta\theta)) \approx \frac{\Delta\theta^2}{2}\int_0^\infty \left(\frac{d}{dx}f(x,\theta)\right)^2dx.
    \end{equation}
\end{theorem}
\begin{proof}
    \begin{equation*}
        \begin{split}
           f(x,\theta+\Delta\theta)&=f(x,\theta)+f'(x,\theta)\Delta\theta+f''(x,\theta)\frac{(\Delta\theta)^2}{2!}+...\\
           f(x,\theta+\Delta\theta)-f(x,\theta)&\approx f'(x,\theta)\Delta\theta+f''(x,\theta)\frac{(\Delta\theta)^2}{2!}\\ 
        \end{split}
    \end{equation*}
    which implies
    \[
    (f(x,\theta+\Delta\theta)-f(x,\theta))^2\approx (\Delta\theta f'(x,\theta))^2. 
    \]
    Then,  
    \[
    \frac{1}{2}\int_0^\infty  (f(x,\theta+\Delta\theta)-f(x,\theta))^2 dx\approx \frac{(\Delta\theta)^2}{2}\int_0^\infty f'(x,\theta)^2 dx.
    \]
 Thus the result follows.
\end{proof}
\begin{corollary}
    Let $X$ follows exponential with the mean $1/\lambda$, \eqref{R4:eqapprox} becomes
    \[    d(f(x,\lambda),f(x,\lambda+\Delta\lambda))\approx \frac{(\Delta\lambda)^2}{4\lambda}
    \]
    It is evident from Figure \ref{fig:example} that the relative extropy value between two exponential densities with parameters $\lambda$ and $\lambda + \Delta \lambda$ increases widely with increase in the $\Delta \lambda$.
\begin{figure}[h!] 
    \centering 
    \includegraphics[width=0.7\textwidth]{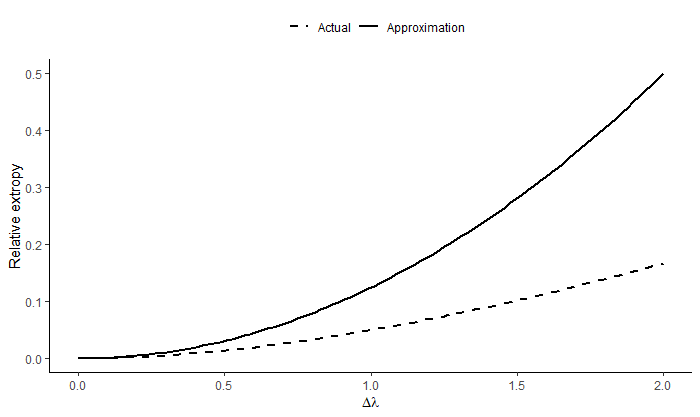} 
    \caption{Comparison of approximation with the actual value of relative extropy for $\lambda=2$. } 
    \label{fig:example} 
\end{figure}
\end{corollary}   
Next, we consider some ordering between two random variables $X$ and $Y$.  We say that, $X$ is said to be smaller (larger) than or equal to $Y$ in the 
\begin{itemize}
    \item[(i)] Extropy ordering, denoted by $X\leq_{ex}(\geq_{ex})Y$ if $J(X)\leq J(Y)$.
    \item[(ii)] Extropy divergence ordering, denoted by $X\leq_{ed} (\geq_{ed})Y$, if $J(f|g)\leq(\geq)J(g|f)$. 
\end{itemize}
\begin{theorem}\label{Theorem:R4relationextropyand divordering}
    $X<_{ex}(>_{ex})Y \iff X>_{ed}(<_{ed})Y$.
\end{theorem}
\begin{proof}
Since $J(f|g)=\xi J(X,Y)-J(X)$ and $J(g|f)=\xi J(X,Y)-J(Y)$, we obtain $J(f|g)-J(g|f)=J(Y)-J(X)$ which gives $J(Y)>(<)J(X)\iff J(f|g) {>(<)}J(g|f)$.
\end{proof}
\begin{corollary}
 Two continuous nonnegative random variables $X$ and $Y$ satisfy the additive extropy model 
 \[
 J(Y) = J(X) +c,
 \]
 if and only if they satisfy the additive extropy divergence model 
 \[
 J(f|g)+c=J(g|f).
 \]
 It also implies $d(f,g)=2J(f|g)+c$.   
\end{corollary}
In the following theorem, we obtain the bounds of extropy divergence based on the orderings.

\begin{theorem}
The following are the conditions for which $J(f|g)$ and $J(g|f)$ are nonnegative.
\begin{itemize}
    \item[(i)] $X>_{ed}Y\iff Y>_{ex}X \implies J(f|g)>0$
    \item [(ii)] $Y>_{ed}X\iff Y<_{ex}X \implies J(g|f)>0$
\end{itemize}
\end{theorem}
\begin{proof}
    Since relative extropy is always nonnegative, it follows that $J(f|g) + J(g|f) > 0$. Given that $J(Y) > J(X)$, we obtain $J(f|g) > J(g|f)$, which implies $J(f|g) - J(g|f) > 0$,  Adding these inequalities results in $2J(f|g) > 0$.  The proof of (ii) follows in a similar way.  
\end{proof}
\section{Relative, inaccuracy and divergence measures of extropy for residual lifetimes}
This section discusses the dynamic scenario of the information measures for residual lifetimes. We focus on deriving some interesting relationships between those information measures. 
The following is the definition of dynamic relative extropy for residual lifetimes.
\begin{definition}
Let $X$ and $Y$ be nonnegative continuous random variables with density functions $f$ and $g$ respectively. Then the dynamic relative extropy is defined as
\begin{equation}\label{R4:eqnDRE}
       d_{r}(f,g,t)=\frac{1}{2}\int_t^\infty \left(\frac{f(x)}{\bar{F}(t)}-\frac{g(x)}{\bar{G}(t)}\right)^2 dx. 
    \end{equation}
\end{definition}
$d_r(f,g,t)$ is a nonnegative symmetric measure that measures the divergence between two residual life distributions.
\begin{example}
    Let $X$ and $Y$ be two exponential random variables with mean $1/\lambda_1$ and $1/\lambda_2$ respectively. Then the dynamic relative extropy $d_r(f,g,t)$ is given by 
    \[
    d_r(f,g,t)=\frac{1}{4}\left(\lambda_1+\lambda_2-\frac{4\lambda_1 \lambda_2}{\lambda_1+\lambda_2}\right),
    \]
    is a constant.
\end{example}
The dynamic extropy divergence between $f$ and $g$ for residual lifetimes (see \cite{mohammadi2024dynamic}, is given by
\[
J_r(f_t|g_t)=\frac{1}{2}\int_t^\infty \left(\frac{f(x)}{\bar{F}(t)}-\frac{g(x)}{\bar{G}(t)}\right)\frac{f(x)}{\bar{F}(t)}=\xi J(X,Y,t)-J_t(X),
\]
and the corresponding dynamic extropy divergence between $g$ and $f$ for residual lifetimes, defined by
\[
J_r(g_t|f_t)=\frac{1}{2}\int_t^\infty \left(\frac{g(x)}{\bar{G}(t)}-\frac{f(x)}{\bar{F}(t)}\right)\frac{g(x)}{\bar{G}(t)}=\xi J(X,Y,t)-J_t(Y).
\]
Analogous to Theorem \ref{Theorem:R4splitRE}, dynamic relative extropy can also be expressed as the sum of dynamic extropy divergences as follows:
    \begin{equation}
         d_{r}(f,g,t)=J_r(f_t|g_t)+J_r(g_t|f_t).
    \end{equation}
Next, we obtain a unique representation of the dynamic extropy inaccuracy \citep{hashempour2024residual} when $X$ is exponential and the hazard function of $Y$ is known.
\begin{theorem}
Let $X$ follows an exponential distribution. The survival extropy inaccuracy $\xi J(X,Y,t)$ is uniquely determined by hazard rate of $Y$ as,
\begin{equation}\label{Eq:R4hazarduniqedynextropy}
\xi J(X,Y,t) = -e^{\lambda t + \int_0^t h_Y(t) dt} \left(\int_t^\infty \frac{\lambda h_Y(t)}{2}e^{-\lambda t - \int_0^t h_Y(t) dt}  dt\right).
\end{equation}
\end{theorem}
\begin{proof}
The survival extropy inaccuracy is given as:
\begin{equation}
\xi J(X,Y,t) = -\frac{1}{2 \bar{F}(t) \bar{G}(t)} \int_t^\infty f(x) g(x) , dx.
\end{equation}
We have $g(x)=\bar{G}(x)h_Y(x)$. Given that $X$ follows an exponential distribution with rate $\lambda$, $h_X(t)=\lambda$. Then
\begin{align}
\xi J(X,Y,t) &= -\frac{1}{2 e^{-\lambda t} \bar{G}(t)} \int_t^\infty \lambda e^{-\lambda x} g(x) dx \
&= -\frac{\lambda}{2 e^{-\lambda t} \bar{G}(t)} \int_t^\infty e^{-\lambda x} \bar{G}(x) h_Y(x)dx.
\end{align}
Since $\bar{G}(x)=e^{-\int_0^t h_Y(x)dx}$, we have
\begin{equation}
\xi J(X,Y,t) = -e^{\lambda t + \int_0^t h_Y(x) dx} \left(\int_t^\infty \frac{\lambda h_Y(t)}{2}e^{-\lambda t - \int_0^x h_Y(u) du}  dx\right).
\end{equation}
\end{proof}
Now we derive the dynamic relative extropy in terms of hazard rate of $Y$ if $X$ is exponential.
\begin{theorem}
Let $X$ follow an exponential distribution with hazard rate $\lambda$. The dynamic relative extropy $d(f,g,t)$ is uniquely determined by the hazard rate of $Y$ and is given by:
\begin{equation}\label{3.7}
\begin{split}
 d_r(f,g,t)= -2e^{\lambda t + \int_0^t h_Y(x) dx} &\left(\int_t^\infty \frac{\lambda h_Y(t)}{2}e^{-\lambda t - \int_0^x h_Y(u) du}  dx\right)\\ 
&+ e^{2\int_0^t h_Y(u) du} \left( \int_t^\infty \frac{h_Y(x)^2}{2} e^{-2\int_0^x h_Y(u)du}\right) dx + \frac{\lambda}{4}.   
\end{split}
\end{equation}
\end{theorem}

\begin{proof}
The dynamic relative extropy is given as:
\begin{equation}\label{3.8}
d(f,g,t) = 2\xi J(X,Y,t) - J_t(Y) - J_t(X).
\end{equation}
Using $g(x) = h_Y(x) \bar{G}(x)$, we get:
\begin{equation}\label{3.9}
J_t(Y) = -\frac{1}{2} e^{2\int_0^t h_Y(u) du} \int_t^\infty h_Y(x)^2 e^{-2\int_0^x h_Y(u)du} dx.
\end{equation}

Since $X$ is exponential with rate $\lambda$:
\begin{equation}\label{3.10}
J_t(X) = -\frac{\lambda}{4}.
\end{equation}
From \eqref{Eq:R4hazarduniqedynextropy}, we get $2\xi J(X,Y,t)$ and substituting it, \eqref{3.9} and \eqref{3.10} in \eqref{3.8}, we obtain the required expression \eqref{3.7}. 
\end{proof}
Similarly, we can prove it for dynamic extropy divergences.
\begin{corollary}
    Let $X$ follow an exponential distribution.  Then the dynamic residual extropy divergence is uniquely determined by $h_Y(x)$ as 
    \[
    J_r(f_t|g_t)=-e^{\lambda t + \int_0^t h_Y(t) dt} \left(\int_t^\infty \frac{\lambda h_Y(t)}{2}e^{-\lambda t - \int_0^t h_Y(t) dt}  dt\right)+\frac{\lambda}{4}.
    \]
\end{corollary}
The following is the relationship with dynamic relative extropy with hazard rates and dynamic extropies of $X$ and $Y$.
\begin{theorem}
    $d_r(f,g,t)$ satisfies the differential equation,
\begin{equation}\label{eq:R4diffeqnrelativeextropy}
        \frac{d}{dt}d_r (f,g,t)-d_r(f,g,t)(h_X(t)+h_Y(t))=(h_Y(t)-h_X(t))(J_t(X)-J_t(Y))-\frac{1}{2}(h_X(t)+h_Y(t))^2.
    \end{equation}
\end{theorem}
\begin{proof}
Differentiating \eqref{R4:eqnDRE} with respect to $t$, we get
\[
\frac{d}{dt}d_r (f,g,t)= h_X(t)h_Y(t)-\frac{h^2_X(t)+h^2_Y(t)}{2}+2\xi J(X,Y,t) (h_X(t)+h_Y(t))-2h_X(t)J_t(X)-2h_Y(t)J_t(Y).
\]
 Since
 \[
 \frac{2h_X(t)h_Y(t)-h^2_X(t)-h^2_Y(t)}{2}=-\frac{(h_X(t)+h_Y(t))^2}{2},
 \]
 and 
 \[
 d_r(f,g,t)=2\xi J(X,Y,t)-J_t(Y)-J_t(X),
 \]
 we have
 \begin{equation*}
 \begin{split}
     \frac{d}{dt}{d_r (f,g,t)}=-\frac{(h_X(t)+h_Y(t))^2}{2}+d_r(f,g,t)(h_X(t)&+h_Y(t))+J_t(X)(h_Y(t)-h_X(t))\\
     +J_t(Y)(h_X(t)-h_Y(t)).
 \end{split}    
 \end{equation*} 
\end{proof}

\begin{corollary}
Using \eqref{eq:R4diffeqnrelativeextropy} and $d_r(f,g,t)$ is nondecreasing then we obtain the lower bound for $d_r(f,g,t)$ in terms of the dynamic extropies and hazard rates of $X$ and $Y$ as
\[
d_r(f,g,t)\geq \frac{(h_X(t)-h_Y(t))}{h_X(t)+h_Y(t)}(J_t(X)-J_t(Y)).
\]    
\end{corollary}
\cite{qiu2018residual} proved that $J_t(X)$ be the dynamic extropy of $X$ and is a constant if and only if $X$ is exponential. Now, we establish characterizations of the exponential distribution using dynamic extropy inaccuracy and dynamic extropy divergence for the residual lifetime.
\begin{theorem}\label{Theorem:R4characterisationdynamicinaccuracyexponential}
  Let  $X$ be an exponential random variable. Then the dynamic extropy inaccuracy $\xi J_r(X,Y,t)$ is a constant if and only if $Y$ is exponential.
\end{theorem}\label{R4:inaccuracyexpcharacterisation}   
\begin{proof}
  $\xi J_r(X,Y,t)$ satisfies the differential equation 
  \[
  \frac{d}{dt}{\xi J_r(X,Y,t)}= \frac{h_X(t)h_Y(t)}{2}+(h_X(t)+h_Y(t))\xi J_r(X,Y,t).
  \]
  Let $\xi J_r(X,Y,t)=c$, a constant and $X$ is exponential with $h_X(t)=\lambda$. We have
  \[
  \frac{\lambda h_Y(t)}{2}=-(\lambda+h_Y(t))c.
  \]
Simplifying, we get 
\[
 h_Y(t)\left(c+\frac{\lambda}{2}\right)=-c \lambda,
\]
implies $Y$ is exponential.
\end{proof}
\begin{lem}
    The dynamic residual extropy divergence satisfies the differential equation
\begin{equation}\label{eq:R4diffeqndynamicresextropydivergence}
    \frac{d}{dt}J_r(f_t|g_t)=(h_X(t)+h_Y(t))J_r(f_t|g_t)+(h_Y(t)-h_X(t))\left(\frac{h_X(t)}{2}+J_t(X)\right).
\end{equation}
\end{lem}
\begin{theorem}
    Let $X$ follow an exponential distribution. Then the dynamic residual extropy divergence $J_r(f_t|g_t)$ is a constant if and only if $Y$ is exponential. 
\end{theorem}
\begin{proof}
Let $X$ follows exponential with $h_X(t)=\lambda$ and $J_t(X)=-\lambda/4$. Also if dynamic extropy divergence $J_r(f_t|g_t)=k$, a constant, then $J_r(f_t|g_t)=0$. Then, \eqref{eq:R4diffeqndynamicresextropydivergence} becomes 
\[
(\lambda+h_Y(t))k+(h_Y(t)-\lambda)\left(\frac{\lambda}{2}-\frac{\lambda}{4}\right)=0.
\]
Rearranging the equation, we obtain
\[
4k{(\lambda+h_Y(t))}=-{\lambda(h_Y(t)-\lambda)},
\]
which gives 
\[
h_Y(t)=\frac{\lambda^2-4k\lambda}{4k+\lambda},
\]
implies that $h_Y(t)$ is constant and $Y$ is exponential. The converse part can be derived directly.
\end{proof}
\begin{theorem}
    Let $J_r(g_t|f_t)$ is a constant. Then if $Y$ follows an exponential distribution then
    \begin{itemize}
        \item[(i)] $\xi J_r(X,Y,t)$ is a constant.
        \item[(ii)] $X$ follows exponential.
        \item [(iii)] $J_r(f_t|g_t)$ is a constant.
        \item[(iv)] $d_r(f,g,t)$ is a constant.   
    \end{itemize}
\end{theorem}
\begin{proof}
 Consider $J_r(g_t|f_t)=a$, a constant in all the following cases. Also, if $Y$ is exponential then $J_t(Y)=c$, a constant.Let us prove each part of the theorem.
 \begin{itemize}
     \item [(i)] Given $J_r(g_t|f_t)=\xi J_r(X,Y,t)-J_t(Y)=a$. If $Y$ follows exponential, then $\xi J_r(X,Y,t)$ will be constant. Under the same condition, if $\xi J_r(X,Y,t)$ is a constant, then $Y$ is exponential.  
     \item [(ii)] Under the given condition, if $Y$ follows exponential distribution then $\xi J_r(X,Y,t)$ is a constant implies $X$ is exponential.
     \item [(iii)]
     \[
     J_r(g_t|f_t)-J_r(f_t|g_t)=J_t(X)-J_t(Y) \implies a-J_r(f_t|g_t)=J_t(X)-c,
     \]
     which further implies $a+c=J_t(X)+J_r(g_t|f_t)$. (ii) gives $X$ is exponential under the given assumptions. Let $J_t(X)=b$, a constant. So we have
     \[
     a+c-b=J_r(f_t|g_t),
     \]
     implies $J_r(f_t|g_t)$ is a constant.
     \item [(iv)] Since we have $d_r(f,g,t)=J_r(f_t|g_t)+J_r(g_t|f_t)$, $d_r(f,g,t)$ is a constant under the given assumptions.
 \end{itemize}
 \end{proof}
The random variable $ X $ is said to be smaller (larger) than or equal to $ Y $ in the 
\begin{enumerate}[(a)]
	\item[(i)] Hazard rate ordering, denoted by $ X \leq _{hr}(\geq _{hr})  Y $, if $ h_X(x) \geq (\leq) h_Y (x) $ for all $ x \geq 0 $,
	\item[(ii)] Dynamic residual extropy ordering, denoted by $ X \leq _{rex} (\geq _{rex}) Y $, if $ J_t(X) \leq (\geq) J_t(Y) $ for all $ x \geq 0 $.
    \item[(iii)] Dynamic residual extropy divergence ordering, denoted by $X \leq_{red} (\geq_{red})Y$, if $J_r(f_t|g_t)\leq (\geq) J_r(g_t|f_t)$ for all $x$. 
\end{enumerate}
The following theorem gives the equivalence of (ii) and (iii).
\begin{theorem}
    $X\leq_{rex} (\geq_{rex})Y \iff X\geq_{red} (\leq)Y.$
\end{theorem}
The following theorem gives the bound for $d_r(f,g,t)$ using hazard rates.
\begin{theorem}
    Let $X$ and $Y$ be two non-negative continuous random variables having CDFs $F$ and $G$, PDFs $f$ and $g$, If $X\leq_{hr}Y$ either $X$ or $Y$ is $DFR$, then 
    \begin{equation}
        \frac{d}{dt}\log d_r(f,g,t)\leq h_X(t)+h_Y(t).
    \end{equation}  
\end{theorem}
\begin{proof}
\cite{toomaj2023extropy} proved that if $X\leq_{hr}Y$ either $X$ or $Y$ is $DFR$, then $J_t(X)\leq J_t(Y)$. So, we have 
\[
(h_Y(t)-h_X(t))(J_t(X)-J_t(Y))\leq 0.
\]
Substituting in \eqref{eq:R4diffeqnrelativeextropy},
\[
\frac{d}{dt}d_r(f,g,t)-d_r(f,g,t)(h_X(t)+h_Y(t))\leq 0,
\]
which implies
\[
\frac{d}{dt}\log d_r(f,g,t)\leq h_X(t)+h_Y(t).
\]
\end{proof}
The following theorem gives a characterization property of $d_r(f,g,t)$.
\begin{theorem}
     The relationship $ \frac{d}{dt}\log d_r(f,g,t)= h_X(t)+h_Y(t)$ holds if and only if $d_r(f,g,t)= (\bar{F}(t)\bar{G}(t))^{-1}$. 
\end{theorem}
\begin{proof}
   The theorem can be easily proved by virtue of the relationship between the hazard rate and the survival function.  
 \end{proof}
\begin{theorem}
The following are the conditions for which $J_r(f_t|g_t)$ and $J_r(g_t|f_t)$ are nonnegative.
\begin{itemize}
    \item[(i)] $X\geq_{red} Y$ $\iff$ $Y\geq_{rex} X \implies J_r(f_t|g_t)\geq 0$.
    \item [(ii)] $Y\geq_{red} X \iff X\geq_{rex} Y \implies J_r(g_t|f_t)\geq 0$.
\end{itemize}
    
\end{theorem}
\begin{proof}
$d_r(f,g,t)$ is always nonnegative and we get
    \begin{equation}\label{eq:R4inequality1}
      J_r(f_t|g_t)+J_r(g_t|f_t)\geq 0.  
    \end{equation}
    Since $J_t(Y)>J_t(X)\iff J_r(f_t|g_t)>J_r(g_t|f_t)$,
   \begin{equation}\label{eq:R4inequality2}
      J_r(f_t|g_t)-J_r(g_t|f_t)=J_t(Y)-J_t(X)\geq 0.
    \end{equation}
    Adding (\ref{eq:R4inequality1}) and (\ref{eq:R4inequality2}), we get $2J_r(f_t|g_t)\geq 0$, which proves (i). Similarly, (ii) is also obtained.
\end{proof}
The strict monotonicity of dynamic relative extropy is explained using hazard rate ordering in the following theorem. 
\begin{theorem}
     Let $f(x)$ and $g(x)$ be the probability density functions corresponding to $X$ and $Y$ respectively, which are strictly decreasing functions in $x$. Then if $X\geq_{hr} Y$, then the dynamic relative extropy is always a strictly increasing function of $t$. 
\end{theorem}
\begin{proof}
    Using $d_r(f,g,t) = J_r(f|g) + J_r(g|f)$, we get
\begin{equation}\label{eq:R4diffd(f,g,t)}
  \frac{d}{dt}d_r(f,g,t)= (h_X(t)+h_Y(t))d_r(f,g,t)+(h_Y(t)-h_X(t))\left(\frac{h_X(t)}{2}+J_t(X)+\frac{h_Y(t)}{2}+J_t(Y)\right)     
    \end{equation}
If $f$ and $g$ are strictly decreasing functions, we have $J_t(X)>-\frac{h_X(t)}{2}$ and $J_t(Y)>-\frac{h_Y(t)}{2}$,  Then, we get
\[
\left(\frac{h_X(t)}{2}+J_t(X)+\frac{h_Y(t)}{2}+J_t(Y)\right)>0.
\]
Also, if $h_Y(t)>h_X(t)$, \eqref{eq:R4diffd(f,g,t)} becomes
\[
\frac{d}{dt}d_r(f,g,t)>(h_X(t)+h_Y(t))d_r(f,g,t)>0,
\]
implies $d_r(f,g,t)$ is strictly increasing function of $t$.
\end{proof}
Next, we derive an upper bound for dynamic relative extropy using hazard rate ordering.
\begin{theorem}
     Let $f(x)$ and $g(x)$ be strictly decreasing functions in $x$.\\ 
     \[
     X\leq_{hr} Y \implies \frac{d}{dt}\log d_r(f,g,t)\leq h_X(t)+h_Y(t).
     \]    
\end{theorem}

\begin{proof}
  Let $f(x)$ and $g(x)$ are strictly decreasing in $x$. For $h_X(t)>h_Y(t)$, (\ref{eq:R4diffd(f,g,t)}) becomes
  \[
  d'_r(f,g,t)<(h_X(t)+h_Y(t))d_r(f,g,t)\implies \frac{d}{dt}\log d_r(f,g,t)<h_X(t)+h_Y(t).
  \]
\end{proof}
\section{Dynamic relative and divergence extropy measures for past lifetimes}
In this section, we focus on past lifetime scenarios of extropy-based information measures and study their interrelationships and properties.
Following is the definition of dynamic past relative extropy.
\begin{definition}
  Let $X_{(t)}$ and $Y_{(t)}$ be two past lifetime random variables, then  dynamic past relative extropy is given by
    \begin{equation}
        d_p(f,g,t)=-\frac{1}{2}\int_0^t\left(\frac{f(x)}{F(t)}-\frac{g(x)}{G(t)}\right)^2 dx.
    \end{equation}   
\end{definition}
$d_p(f,g,t)$ is always nonnegative and symmetric with respect to $f$ and $g$. It is equal to zero if and only if $f(x)=g(x)$ for every $x$. The relationship between $\xi J_p(X,Y,t)$, $J(_t X)$, and $J(_t Y)$ is given by 
\[
d_p(f,g,t)=2\xi J(X,Y,t)-J_t(X)-J_t(Y).
\]

    Dynamic past relative extropy is the sum of the corresponding dynamic past extropy divergences. i.e,
    \[
    d_p(f,g,t)= J_{p}(f_t|g_t)+ J_{p}(g_t|f_t),
    \]
  where
\[
J_{p}(g_t|f_t)=-\frac{1}{2}\int_0^t \left(\frac{g(x)}{G(t)}-\frac{f(x)}{F(t)}\right)\frac{g(x)}{G(t)}dx,
\] 
and $ J_{p}(f_t|g_t)$ can be represented using inaccuracy and extropy as follows
\[
 J^{p}_t(f|g)=\xi J_p(X,Y,t)-J(_t X).
\]
\begin{example}
    Let $X$ and $Y$ follow distribution derived by \cite{unnikrishnan2021some} with $F(x)=e^a(x-b)$ and $G(y)=e^c(y-d)$ respectively. The failure extropy of $X$ and $Y$ at $t$ is given below.
    \begin{equation}
        J(_t X)=-\frac{1}{2}\int_0^t \frac{f^2 (x)}{F^2 (t)}dx=\frac{1}{-2e^{2at}}\left(1+\frac{a}{2}(e^{2at}-1)\right).
    \end{equation}
The dynamic past inaccuracy measure is also a decreasing function of $t$ and is given by
\begin{equation}
    \xi J_p(X,Y,t)=-\frac{1}{2e^{(a+c)t}}\left(1+\frac{ac}{a+c}(e^{t(a+c)}-1)\right).
\end{equation}
Now, the dynamic past extropy divergence can be derived as 
\begin{equation}
J_{p}(f_t|g_t)=\frac{1}{2e^{2at}}\left(1+\frac{a}{2}(e^{2at}-1)\right)-\frac{1}{2e^{(a+c)t}}\left(1+\frac{ac}{a+c}(e^{t(a+c)}-1)\right).
\end{equation}
The dynamic relative extropy of $X$ and $Y$ is given by 
\begin{equation}
d_p(f,g) = -\frac{1}{e^{(a+c)t}} \left(1 + \frac{ac}{a+c} \left( e^{t(a+c)} - 1 \right)\right) 
+ \frac{1}{2e^{2at}} \left(1 + \frac{a}{2} \left( e^{2at} - 1 \right)\right)
+ \frac{1}{2e^{2ct}} \left(1 + \frac{c}{2} \left( e^{2ct} - 1 \right)\right).
\end{equation}
\end{example}
The random variable $ X $ is said to be smaller (larger) than or equal to $ Y $ in the 
\begin{enumerate}[(a)]
	\item[(i)] Reversed hazard rate ordering, denoted by $ X \leq _{rh}(\geq _{rh})  Y $, if $ \lambda_X(x) \geq (\leq) \lambda_Y (x) $ for all $ x \geq 0 $,
	\item[(ii)] Dynamic past extropy ordering, denoted by $ X \leq _{pex} (\geq _{pex}) Y $, if $ J(_t X) \leq (\geq) J(_t Y) $ for all $ x \geq 0 $.
    \item[(iii)] Dynamic past extropy divergence ordering, denoted by $X \leq_{ped} (\geq_{ped})Y$, if $J_p(f_t|g_t)\leq (\geq) J_p(g_t|f_t)$ for all $x$. 
\end{enumerate}
The following theorem gives the relationship between (ii) and (iii). 
\begin{theorem}
    $X\leq_{pex} (\geq_{pex})Y \iff X\geq_{ped} (\leq_{ped})Y.$
\end{theorem}
The $J_{p}(f_t|g_t)$ and $J_{p}(g_t|f_t)$ can have both nonnegative and nonpositive values. The following theorem gives the conditions for which they are always positive.
\begin{theorem}
The following are the conditions for which $J_{p}(f_t|g_t)$ and $J_{p}(g_t|f_t)$ are always positive.
\begin{itemize}
    \item[(i)] If $X<_{ped}Y$(or if $X>_{pex}Y$) for all $t$, then $J_{p}(g_t|f_t)>0$.
    \item [(ii)] If  $Y<_{ped}X$ (or if $X<_{pex}Y$) for all $t$, then $J_{p}(f_t|g_t)>0$.
\end{itemize}   
\end{theorem}
Now, we derive a relationship between extropy divergence and its dynamic forms.
\begin{theorem}
    $J(f|g),J_r(f_t|g_t)$, and $J_p(g_t|f_t)$ holds the following relationship:
    \begin{equation}
        J(f|g)=\bar{F}(t)\bar{G}(t)J_r(f_t|g_t)+F(t)G(t)J_p(g_t|f_t)+(\bar{G}(t)-\bar{F}(t))(J_t(X)-J(_t X)).
    \end{equation}
\end{theorem}
\begin{proof}
The relationship between extropy inaccuracy and its dynamic forms is given by \cite{mohammadi2024dynamic} as follows:
    \begin{equation}
        \xi J(X,Y)=F(t)G(t) \xi J_p(X,Y,t)+\bar{F}(t)\bar{G}(t)\xi J_r(X,Y,t).
    \end{equation}
    We have $J(X)=\bar{F}^2(t)J_t(X)+F^2(t)J(_t X)$, applying in $J(f|g)= \xi J(X,Y)-J(X)$, we get
    \begin{align*}
        \xi J(X,Y)-J(X)&=F(t)G(t)\xi J_p(X,Y,t)+\bar{F}(t)\bar{G}(t)\xi J_r(X,Y,t)-\bar{F}^2(t)J_t(X)-F^2(t)J(_t X)\\
        &=\bar{F}(t)(\bar{G}(t)\xi J_r(X,Y,t)-\bar{F}(t)J_t(X))+F(t)(G(t)\xi J_p(X,Y,t)-F(t)J(_t X)).
    \end{align*}
    Substituting $J_r(f_t|g_t)+J_t(X)=\xi J_r(X,Y,t)$ and $J_p(f_t|g_t)+J(_t X)=\xi J_p(X,Y,t)$, we obtain
    \begin{equation*}
        J(f|g)=\bar{F}(t)(\bar{G}(t)J_r(f_t|g_t)+J_t(X)(\bar{G}(t)-\bar{F}(t))+F(t)(G(t)J_p(f_t|g_t)+J(_t X)(G(t)-F(t)).
    \end{equation*}
    We have $G(t)-F(t)=\bar{F}(t)-\bar{G}(t)$. So,
    \begin{equation*}
      J(f|g)= \bar{F}(t)(\bar{G}(t)J_r(f_t|g_t)+ F(t)(G(t)J_p(f_t|g_t)+(\bar{G}(t)-\bar{F}(t))(J_t(X)-J(_t X)).
    \end{equation*}
\end{proof}

The following is a characterization of the distribution having constant reversed hazard rate distribution defined by \cite{unnikrishnan2021some} with distribution function:
    \begin{equation}\label{eq:R4ConstantrevhazarddistbnY}
     G(x)=e^{c(x-d)},0\leq x \leq d, c>0.   
    \end{equation}
\begin{theorem}
Let $X$ follows the distribution function of the form (\ref{eq:R4ConstantrevhazarddistbnY}) with 
\begin{equation}\label{eq:R4ConstantrevhazarddistbnX}
    F(x)=e^{a(x-b)},0\leq x \leq b, a>0. 
\end{equation}
Then $\xi J_p(X,Y,t)$ be a constant if and only if $Y$ follows (\ref{eq:R4ConstantrevhazarddistbnY}).  
\end{theorem}
\begin{proof}
  The differential equation of $\xi J_p(X,Y,t)$ is given by
  \begin{equation}\label{eq:R4residinaccuracydiffequation}
      \xi J'_p(X,Y,t)+\xi J_p(X,Y,t)(\lambda_X(t)+\lambda_Y(t))=\frac{-\lambda_X(t)\lambda_Y(t)}{2}.
  \end{equation}
  Assume that $\xi J_p(X,Y,t)=k$, a constant, which implies $\xi J'_p(X,Y,t)=0$.
Now, (\ref{eq:R4residinaccuracydiffequation}) becomes
\[
k(\lambda_X(t)+\lambda_Y(t))=\frac{-\lambda_X(t)\lambda_Y(t)}{2},
\]
implies
\[
\frac{1}{\lambda_Y(t)}+\frac{1}{\lambda_X(t)}=-\frac{1}{2c}.
\]
 If $X$ follows (\ref{eq:R4ConstantrevhazarddistbnX}) with constant reversed hazard rate :
  \[
  \lambda(x)=
  \begin{cases} 
      1 & x=0 \\
      a & 0<x\leq b\\
   \end{cases}
\]
then, $\lambda_Y(t)$ is also a constant and vice versa.
\end{proof}
The following theorem gives the relationship between relative extropy with its dynamic forms.
\begin{theorem}
    $d(f,g)$,$d_p(f,g,t)$, and $d_r(f,g,t)$ satisfies the following relationship.
    \begin{equation}
    \begin{split}
        d(f,g)=d_p(f,g,t)&F(t)G(t)+d_r(f,g,t)\bar{F}(t)\bar{G}(t)+(\bar{F}(t)-\bar{G}(t))\times\\
        &\left(\bar{G}(t)J_t(Y)+F(t)J(_t X)-\bar{F}(t)J_t(X)-G(t)J(_t Y)\right).
        \end{split}
    \end{equation}
\end{theorem}

\section{Estimation and application of relative extropy}
In this section, we estimate the relative extropy function and demonstrate some real-data applications of the measure.  Let $X_1, X_2, X_3, \cdots, X_n$ be a random sample drawn from a population that has a distribution function $F$ and $Y_1, Y_2, Y_3, \cdots, Y_n$ be a random sample drawn from a population that has a distribution function $G$. We propose a nonparametric estimator for the relative extropy using kernel density estimation.	We assume that the kernel $k(.)$ satisfies the following conditions.
\begin{itemize}
	\item $k(x)\geq 0$ for all $x$.
	\item $\int k(x)dx=1.$
	\item $k(.)$ is symmetric about zero.
	\item $k(.)$ satisfies Lipschitz condition, namely there exists a constant $M$ such that \\$|k(x)-k(y)|\leq M|x-y|.$
\end{itemize}
The kernel density estimator of $f(x)$ is given by \cite{parzen1962estimation}
\begin{equation*}
	\hat{f}_n(x)=\frac{1}{nb_n}\sum_{j=1}^{n}k\left(\frac{x-X_j}{b_n}\right),
\end{equation*}
where $k(\cdot)$ is a kernel of order $s$ with compact support and ${b_n}$, the bandwidths, is a sequence of positive numbers such that
$b_n\rightarrow 0$ and $nb_n\rightarrow\infty$ as $n\rightarrow \infty$.
 Based on this we define the nonparametric kernel estimator of $d(f,g)$ as
\begin{equation}\label{R4:eqREestimate}
\begin{split}
\hat{d}(\hat{{f}}_n,\hat{{g}}_n)=& \frac{1}{2}\int_0^{\infty}(\hat{{f}}_n(x)-\hat{{g}}_n(x))^2 dx.\\ 
\end{split}
\end{equation}
 We use the Gaussian kernel function since it is the most frequently used and produced the smoothest estimate among other kernel functions. Bandwidth $B_n$ was chosen using the Sheather and Jones method (see \cite{sheather1991reliable}).
\begin{table}[]
    \centering
    \begin{tabular}{p{2cm}|p{2cm}|p{2cm}|p{3cm}}
    \hline
       $n$  & $\hat{d}(f,g)$ & Bias & MSE  \\
       \hline
        50 & 0.06665 & 0.01665 & 0.00027\\
        75 & 0.07234 & 0.01096 & 0.00012 \\
        100 & 0.08364 & 0.00034 & $1.156\times 10^{-7}$\\
        \hline
    \end{tabular}
    \caption{Bias and MSE of two exponential distributions with $\lambda_1=1$ and $\lambda_2=2$ with actual value 0.0833}
    \label{Table:R4relativeextropyexponential1}
\end{table}
\begin{table}[]
    \centering
    \begin{tabular}{p{2cm}|p{2cm}|p{2cm}|p{3cm}}
    \hline
       $n$  & $\hat{d}(f,g)$ & Bias & MSE  \\
       \hline
        50 & 0.27775 & 0.04368 & 0.001907\\
        75 & 0.36161 & 0.04018 & 0.001614 \\
        100 & 0.08364 & 0.01129 & 0.00013\\
        \hline
    \end{tabular}
    \caption{Bias and MSE of two exponential distributions with $\lambda_X=2$ and $\lambda_Y=5$ with actual value 0.32143}
    \label{Table:R4relativeextropyexponential1}
\end{table}
\begin{table}[h]
    \centering
    \begin{tabular}{p{2cm}|p{2cm}|p{2cm}|p{3cm}}
    \hline
       $n$  & $\hat{d}(f,g)$ & Bias & MSE  \\
       \hline
        50 & 0.02899 & 0.00515 & $2.6523\times 10^{-5}$\\
        75 & 0.029917 & 0.00422 &$1.78337\times 10^{-5}$ \\
        100 & 0.03579 & 0.00165 & $2.7225\times 10^{-6}$\\
        \hline
    \end{tabular}
    \caption{Bias and MSE of two Weibull distributions ($k_X=1.5$,$\lambda_X=2$,$k_Y=2$,$\lambda_Y=3$) with actual value 0.03414}
    \label{Table:R4relativeextropyexponential1}
\end{table}
To illustrate the usefulness of the proposed relative survival extropy, we consider the data set from \cite{lagakos1981case} that contains the observed time to death and causes of death for 400 mice exposed to varying doses of Red Dye No. 40 in a lifetime feeding experiment aimed at evaluating its carcinogenicity.  Using relative extropy, we analyze the impact of time to death on gender and tumor presence. Our study examines the differences in lifetime distributions between male and female mice, both with and without tumors. 
\begin{table}[h!]
\begin{center}
\caption{Relative extropy for female and male mices}
\label{Table:R4relatveextropyofmice} 
 \begin{tabular}{p{4cm}|p{4cm}|p{3cm}}
 \hline
 $X$ & $Y$ & $\hat{d}(f,g)\times 10^{-4}$ \\ 
 \hline
 Male  & Female & 2.995463 \\ 
 Male(tumour) & Female (tumour) & 6.527382\\
 Male & Male (tumour) & 4.54618\\
 Female & Female (tumour) & 8.682015 \\
 Female(tumour) & Male & 4.189252\\
 Female & Male(tumour) & 13.4785\\
 \hline
\end{tabular}
\end{center}
\end{table}
\begin{figure}[h!]
    \centering
    \includegraphics[width=0.6\linewidth]{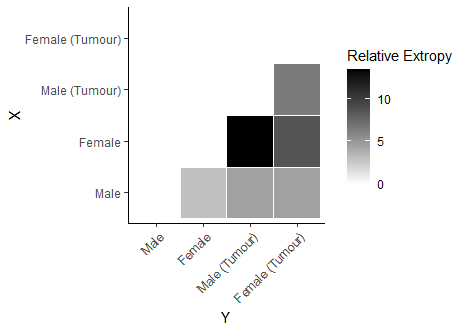}
    \caption{Heatmapofmicedata}
    \label{fig:R4heatmapofmicedata}
\end{figure}
\ref{fig:R4heatmapofmicedata} shows the graphical representation of the relative extropy of time to death of mice.
Table \ref{Table:R4relatveextropyofmice} gives the estimated relative extropy values, which quantify the differences in lifetime distributions between various groups of mice according to gender and tumor presence. Higher values indicate greater dissimilarity in the survival patterns of the compared groups. The presence of Tumors leads to significant changes in survival patterns, particularly in females, as indicated by the highest relative extropy value. Gender differences are more evident when tumors are present, as seen in the male vs. female tumor-affected comparison. Males and females show some survival differences even without tumors, but these differences are smaller compared to those influenced by tumor presence.

Next, we use a mall customer data set (see \url{https://www.kaggle.com/datasets/mall-customer-dataset}), which typically includes information about 200 individual mall shoppers. This usually contains details like their unique customer ID, gender, age, annual income, and a "spending score" that reflects their purchasing behavior and spending patterns within the mall. This allows for customer segmentation and analysis to understand shopping trends and effectively target specific demographics. 
We estimate the difference in patterns of spending score based on gender, annual income, and age of customers. The relative extropy between the pattern of spending score between Male and Female customers is 0.0004891622, which is relatively low. 
\begin{table}[h!]
    \centering
    \begin{tabular}{cc|c}
        \toprule
        \textbf{Income Group 1} & \textbf{Income Group 2} & \textbf{Relative Extropy of Spending scores} \\
        \midrule
        {[15,37.8]}     & (37.8,54]     & 0.01694059  \\
        {[15,37.8]}     & (54,67]       & 0.02204038  \\
        {[15,37.8]}     & (67,78.2]     & 0.0001550714 \\
        {[15,37.8]}     & (78.2,137]    & 0.0006134496 \\
        (37.8,54]       & (54,67]       & 0.001504017  \\
        (37.8,54]       & (67,78.2]     & 0.01689823  \\
        (37.8,54]       & (78.2,137]    & 0.01845978  \\
        (54,67]         & (67,78.2]     & 0.02203957  \\
        (54,67]         & (78.2,137]    & 0.02389467  \\
        (67,78.2]       & (78.2,137]    & 0.0003740169 \\
        \bottomrule
    \end{tabular}
    \caption{Relative Extropy between Spending Scores of Different Income Groups}
    \label{Table:R4customer_spending_income_RE}
\end{table}
\begin{figure}
    \centering
    \includegraphics[width=0.6\linewidth]{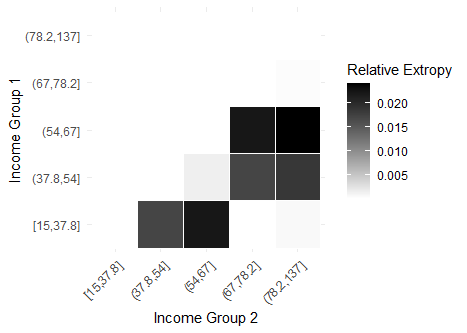}
    \caption{Heatmap of relative extropy between income groups.}
    \label{fig:R4heatmapREincomegroups}
\end{figure}

The analysis of relative extropy in spending scores across income groups reveals both similarities and dissimilarities (see Table  \ref{Table:R4customer_spending_income_RE}). The income is given in $1000\$$ (1000 dollars). We have divided customers into 5 income groups based on 4 quantile values (20 percentile, 40 percentile, 60 percentile, 80 percentile) and calculated the estimate of the relative extropy of spending scores for all possible income group combinations. The income groups based on 4 quantile values using the given data set are $[15,37.8],(37.8,54],(54,67],(67,78.2]$, and $(78.2,137]$. Higher relative extropy implies high dissimilarity in spending score patterns. Figure \ref{fig:R4heatmapREincomegroups} shows the heat map showing the relative extropy between all pairs of income groups. The highest relative extropy (0.02389467) occurs between the (54,67] and (78.2,137] income groups, indicating significant differences in spending patterns. Similarly, the (15,37.8] and (54,67] groups also show notable variation (0.02204038), suggesting that spending scores shift considerably as income increases. In contrast, the lowest relative extropy (0.0001550714) is observed between (15,37.8] and (67,78.2], implying highly similar spending patterns. Another low divergence (0.0003740169) between (67,78.2] and (78.2,137] suggests stability in spending behavior among higher-income individuals. Overall, while spending habits evolve with income, the most substantial behavioral shifts occur in the middle-income range, whereas lower and higher-income groups exhibit more similar and stable spending patterns.
\begin{table}[h]
    \centering
    \begin{tabular}{cc|c}
        \toprule
        \textbf{Age Group 1} & \textbf{Age Group 2} & \textbf{Relative Extropy of Spending Scores} \\
        \midrule
        {[18,26.8]}     & (26.8,32]     & 0.001615878  \\
        {[18,26.8]}     & (32,40]       & 0.001316756  \\
        {[18,26.8]}     & (40,50.2]     & 0.004239812  \\
        {[18,26.8]}     & (50.2,70]     & 0.003509331  \\
        (26.8,32]       & (32,40]       & 0.002229807  \\
        (26.8,32]       & (40,50.2]     & 0.009462067  \\
        (26.8,32]       & (50.2,70]     & 0.009404873  \\
        (32,40]         & (40,50.2]     & 0.004235115  \\
        (32,40]         & (50.2,70]     & 0.005640369  \\
        (40,50.2]       & (50.2,70]     & 0.00151164  \\
        \bottomrule
    \end{tabular}
    \caption{Relative Extropy between Spending Scores of Different Age Groups}
    \label{Table:R4customer_spending_age_RE}
\end{table}

Next, we estimate similarities and differences in spending behavior based on different age groups (see Table \ref{Table:R4customer_spending_age_RE}). The grouping is based on 4 quantiles so we have 5 age groups of customers. They are [18,26.8],(26.8,32],(32,40],(40,50.2] and (50.2,70]. Figure \ref{fig:R4heatmapREagegroups} shows the heatmap of relative extropy between different age groups. The highest divergence (0.009462067) is observed between the (26.8,32] and (40,50.2] age groups, suggesting a significant shift in spending habits as individuals move from early adulthood to middle age. A similarly high divergence (0.009404873) between (26.8,32] and (50.2,70] indicates that younger adults and older individuals have distinct financial priorities. 
The findings indicate that spending behavior evolves gradually with age, with the most significant shifts occurring as individuals transition from their late 20s to middle age.
\begin{figure}
    \centering
    \includegraphics[width=0.6\linewidth]{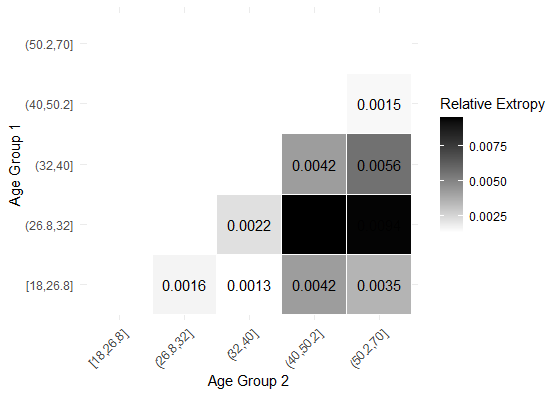}
    \caption{Heatmap of relative extropy between age groups.}
    \label{fig:R4heatmapREagegroups}
\end{figure}
\section*{Acknowledgements}
The first author would like to thank Cochin University of Science and Technology, India, for the financial support.

\section*{Conflict of Interest statement}
On behalf of all authors, the corresponding author declares that there is no conflict of interest.
\section*{Data availability statement}

\bibliographystyle{apalike}
\bibliography{reference}

\begin{thebibliography}{}

\bibitem[Hashempour and Mohammadi, 2024a]{hashempour2024new}
Hashempour, M. and Mohammadi, M. (2024a).
\newblock A new measure of inaccuracy for record statistics based on extropy.
\newblock {\em Probability in the Engineering and Informational Sciences}, 38(1):207--225.

\bibitem[Hashempour and Mohammadi, 2024b]{hashempour2024dynamic}
Hashempour, M. and Mohammadi, M. (2024b).
\newblock On dynamic cumulative past inaccuracy measure based on extropy.
\newblock {\em Communications in Statistics-Theory and Methods}, 53(4):1294--1311.

\bibitem[Hashempour et~al., 2024]{hashempour2024residual}
Hashempour, M., Toomaj, A., and Kazemi, M.~R. (2024).
\newblock Residual inaccuracy extropy and its properties.
\newblock {\em Mathematica Slovaca}, 74(5):1321--1342.

\bibitem[Kerridge, 1961]{kerridge1961inaccuracy}
Kerridge, D.~F. (1961).
\newblock Inaccuracy and inference.
\newblock {\em Journal of the Royal Statistical Society. Series B (Methodological)}, 23(1):184--194.

\bibitem[Kullback, 1959]{kullback1959}
Kullback, S. (1959).
\newblock {\em Information Theory and Statistics}.
\newblock Wiley, NY.

\bibitem[Kullback and Leibler, 1951]{kullback1951information}
Kullback, S. and Leibler, R.~A. (1951).
\newblock On information and sufficiency.
\newblock {\em The Annals of Mathematical Statistics}, 22(1):79--86.

\bibitem[Lad et~al., 2015]{lad2015extropy}
Lad, F., Sanfilippo, G., and Agro, G. (2015).
\newblock Extropy: Complementary dual of entropy.
\newblock {\em Statistical Science}, 30(1):40--58.

\bibitem[Lagakos and Mosteller, 1981]{lagakos1981case}
Lagakos, S. and Mosteller, F. (1981).
\newblock A case study of statistics in the regulatory process: the fd\&c red no. 40 experiment.
\newblock {\em Journal of the National Cancer Institute}, 66(1):197--212.

\bibitem[Mohammadi et~al., 2024]{mohammadi2024dynamic}
Mohammadi, M., Hashempour, M., and Kamari, O. (2024).
\newblock On the dynamic residual measure of inaccuracy based on extropy in order statistics.
\newblock {\em Probability in the Engineering and Informational Sciences}, pages 1--22.

\bibitem[Nair et~al., 2021]{unnikrishnan2021some}
Nair, N.~U., Sunoj, S.~M., and Rajesh, G. (2021).
\newblock Some aspects of reversed hazard rate and past entropy.
\newblock {\em Communications in Statistics-Theory and Methods}, 50(9):2106--2116.

\bibitem[Parzen, 1962]{parzen1962estimation}
Parzen, E. (1962).
\newblock On estimation of a probability density function and mode.
\newblock {\em The Annals of Mathematical Statistics}, 33(3):1065--1076.

\bibitem[Qiu and Jia, 2018]{qiu2018residual}
Qiu, G. and Jia, K. (2018).
\newblock The residual extropy of order statistics.
\newblock {\em Statistics \& Probability Letters}, 133:15--22.

\bibitem[Saranya and Sunoj, 2025]{saranya2025inaccuracy}
Saranya, P. and Sunoj, S. (2025).
\newblock Inaccuracy and divergence measures based on survival extropy with applications in testing and image analysis.
\newblock {\em Japanese Journal of Statistics and Data Science}, pages 1--27.

\bibitem[Saranya and Sunoj, 2024]{saranya2024relative}
Saranya, P. and Sunoj, S.~M. (2024).
\newblock On relative cumulative extropy, its residual (past) measures and their applications in estimation and testing.
\newblock {\em Journal of the Indian Society for Probability and Statistics}, 25:199--225.

\bibitem[Shannon, 1948]{shannon1948mathematical}
Shannon, C.~E. (1948).
\newblock A mathematical theory of communication.
\newblock {\em The Bell System Technical Journal}, 27(3):379--423.

\bibitem[Sheather and Jones, 1991]{sheather1991reliable}
Sheather, S.~J. and Jones, M.~C. (1991).
\newblock A reliable data-based bandwidth selection method for kernel density estimation.
\newblock {\em Journal of the Royal Statistical Society: Series B (Methodological)}, 53(3):683--690.

\bibitem[Toomaj et~al., 2023]{toomaj2023extropy}
Toomaj, A., Hashempour, M., and Balakrishnan, N. (2023).
\newblock Extropy: Characterizations and dynamic versions.
\newblock {\em Journal of Applied Probability}, 60(4):1333--1351.

\end{thebibliography}
\end{document}